\documentclass{ieeeaccess}
\usepackage{cite}
\usepackage{amsmath,amssymb,amsfonts}
\usepackage{algorithmic}
\usepackage{graphicx}
\usepackage{comment}
\usepackage{textcomp}
\usepackage{color}
\usepackage{ulem}
\def\BibTeX{{\rm B\kern-.05em{\sc i\kern-.025em b}\kern-.08em T\kern-.1667em\lower.7ex\hbox{E}\kern-.125emX}}
\usepackage{caption,setspace}
\usepackage{epsfig} 
\usepackage{algorithm2e}
\usepackage{dsfont}
\usepackage{bbm}

\usepackage{amsmath,amssymb,amsthm}

\usepackage{textcomp}
\UseRawInputEncoding

\theoremstyle{theorem}
\newtheorem{theorem}{Theorem}%[section]
\newtheorem{lemma}{Lemma}
\newtheorem*{cor}{Corollary}
\theoremstyle{definition}
\newtheorem{definition}{Definition}%[section]

\theoremstyle{remark}

\begin{document}
	\history{Date of publication xxxx 00, 0000, date of current version xxxx 00, 0000.}
	\doi{-}
	\title{Full-Resilient Memory-Optimum Multi-Party Non-Interactive Key Exchange}
	\author{\uppercase{Majid Salimi}\authorrefmark{1}, %\IEEEmembership{Fellow, IEEE},
		\uppercase{Hamid Mala\authorrefmark{2}, Honorio Martin \authorrefmark{3}, and Pedro Peris-Lopez. \authorrefmark{4}}
		%\IEEEmembership{Member, IEEE}
	}
	\address[1]{
		Faculty of Computer Engineering, University of Isfahan, Isfahan, Iran. (e-mail: M.Salimi@eng.ui.ac.ir )}
	\address[2]{
		Faculty of Computer Engineering, University of Isfahan, Isfahan, Iran. (e-mail: h.mala@eng.ui.ac.ir)}
	\address[3]{ Department of Electronic Technology, University Carlos III of Madrid, Avda. de la Universidad 30, 28911 Leganés, Spain (e-mail: hmartin@ing.uc3m.es)}
	\address[4]{
		Department of Computer Science, University Carlos III of Madrid, Avda. de la Universidad 30, 28911 Leganés, Spain. (e-mail:pperis@inf.uc3m.es)}
	\tfootnote{}
	\markboth
	{Majid Salimi \headeretal: Full-Resilient Memory-Optimum Multi-party Non-Interactive Key Exchange}
	{Majid Salimi et al. \headeretal: Full-Resilient Memory-Optimum Multi-party Non-Interactive Key Exchange}
	
	\corresp{Corresponding author: Hamid Mala (e-mail:h.mala@eng.ui.ac.ir).}

	\begin{abstract}
		Multi-Party Non-Interactive Key Exchange (MP-NIKE) is a fundamental cryptographic primitive in which users register into a key generation centre and receive a public/private key pair each. After that, any subset of these users can compute a shared key without any interaction. Nowadays, IoT devices suffer from a high number and large size of messages exchanged in the Key Management Protocol (KMP). To overcome this, an MP-NIKE scheme can eliminate the airtime and latency of messages transferred between IoT devices. 
		
		MP-NIKE schemes can be realized by using multilinear maps. There are several attempts for constructing multilinear maps based on indistinguishable obfuscation, lattices and the Chinese Remainder Theorem (CRT). Nevertheless, these schemes are inefficient in terms of computation cost and memory overhead. Besides, several attacks have been recently reported against CRT-based and lattice-based multilinear maps. There is only one modular exponentiation-based MP-NIKE scheme in the literature which has been claimed to be both secure and efficient. In this article, we present an attack on this scheme based on the Euclidean algorithm, in which two colluding users can obtain the shared key of any arbitrary subgroup of users. We also propose an efficient and secure MP-NIKE scheme. We show how our proposal is secure in the random oracle model assuming the hardness of the root extraction modulo a composite number.
	\end{abstract}
	
	\begin{IEEEkeywords}
		Multi-Party Non-Interactive Key Exchange, Broadcast Encryption, Internet of Things, Random Oracle Model.
	\end{IEEEkeywords}

	\maketitle
	\section{Introduction}
	\label{intro}
	In a key distribution scheme, an off-line Key Generation Center (KGC) distributes keying information through a secure channel to every node (user) in the network. Later, every pair of users in the system, by using the keying information they hold, will be able to determine a key known only to them. This operating mode enables them to have encrypted communications \cite{Stinson}. \color{black} Suppose we have a set of $n$ nodes. \color{black} In its general form, called the multi-party scenario, the key distribution problem is not restricted to only pairs of users, but it must enable any arbitrary subset of these $n$ nodes to determine a shared key \cite{Zhang2018}. A trivial solution to this problem is that a Trusted Authority (TA) generates $M=2^n-n-1$ symmetric keys, and assigns each to one of the $M$ subsets with at least two members. Then, it gives the key for each group (subset of users) to the users who belong to this subset. Any node is a member of $G=2^{n-1}-1$ groups of at least two members. As a result, any node must store $G$ distinct keys, which is impractical. Public-key cryptographic approaches can be employed to address this limitation \cite{LiYuMin2018}. When, instead of a pool of symmetric keys, any user receives only one public/private key pair from the KGC and employs its private key and other users' public keys to generate a shared symmetric key (without any interaction), the scheme is usually referred to as Non-Interactive Key Exchange (NIKE) \cite{NIKE}.\\ %Besides, if the size of subsets is not limited, then the approach will be multi-party. \\
	
	In this article, we focus on multi-party solutions. Notably, in Multi-Party Non-Interactive Key Exchange (MP-NIKE)\- schemes, any user first registers into a KGC and receives a unique public/\-private key pair. Let $W$ be a subset of registered users. Then any user $U_i \in W$ can compute a pre-shared key $K_W$ using its private key and the public keys of other members of the group $W$.
	
	\color{black} We stress that in this article, we focus on key distribution not on the key agreement schemes and their associated features such as forward secrecy or authentication such as in \cite{SPKAE2}, \cite{authen2019} and \cite{authen2019b}. Nevertheless, in our scheme, by using the public key of each node as node identifier, nodes can be authenticated. Because no one else has access to the respective private key and without a valid private key, no one can compute the shared key, so we can be sure that nodes are authenticated. The proposed approach provides a long-term key, not a session key. The long-term key can then be employed as a symmetric key pre-shared among the nodes of the associated subgroup to run  an authentication and/or session key agreement  protocol.  
	\color{black}
	
	\subsection{Applications of MP-NIKE}
	\color{black}As an underlying cryptographic protocol, NIKE has many applications including broadcast encryption, \color{black}key management for wireless sensor networks (WSNs) \cite{Athmani2019,Amin2018}, and group communication for Internet of things (IoT) devices \cite{Guo2018,Mahmood2018}.\\
	
	\textbf{Group Communication for IoT.} One of the applications of MP-NIKE schemes is key management for group communication in Internet of Thing. Suppose a few smart objects in a smart home need to securely communicate, so first they need to securely establish a session key. The key agreement protocols for IoT, such as  \cite{IEEE3}, \cite{IEEE6}, \cite{IEEE7}, \cite{IEEEAccess1}, \cite{IEEEAccess3}, \cite{DingWang2019} and \cite{IEEEAccess4} need to exchange a few messages, while  message exchanging is costly and time-consuming. This is a serious challenge in resource-constrained devices employed in IoT systems; for example, the Maximum Transmission Unit (MTU) at the link layer of Industrial IoT (IIoT) technology, when IEEE 802.15.4 technology is adopted, is  just equal to 127 bytes, so the Key Management Protocol (KMP) messages must be fragmented \cite{X509}. Therefore, even one message can cause too airtime latency. For example, two-party key agreement by running KMP protocol with implicit X.509 certificates takes 3.29 seconds (computation time plus air latency) in a single-hop network \cite{X509}. It gets even worse in multi-hop networks \cite{X509}. 
	
	An MP-NIKE scheme can provide a distinct key between any subset of these objects and without needing to exchange any messages. Therefore by using MP-NIKE, the IoT nodes are not limited in two-way connections, and can securely establish a shared key between any arbitrary group of nodes. Furthermore, by employing MP-NIKE, nodes can securely broadcast messages for other nodes. The new IoT devices can natively and simply compute cryptographic primitives \cite{X509}. Furthermore, as we will compare \color{black}in Section \ref{performance}, \color{black}the proposed MP-NIKE scheme is practical and efficient.\\

	\begin{figure}
		\label{Pic1}
		
		\center 
		\includegraphics[width=3.5	in,height=4.5in,clip,keepaspectratio]{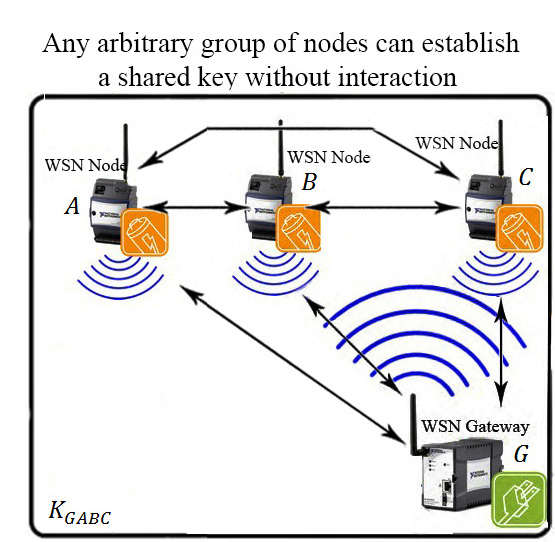}
		\caption{MP-NIKE in WSN} 
	\end{figure}
	\begin{figure}
		\label{Pic2}
		
		\center 
		\includegraphics[width=3.5	in,height=3.2in,clip,keepaspectratio]{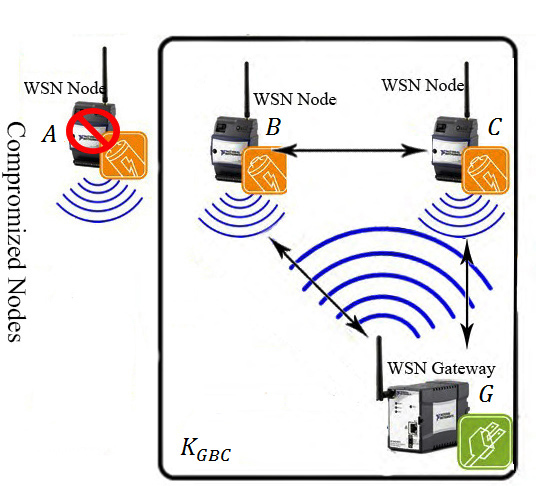}
		\caption{Removing the compromised node in WSN} 
	\end{figure}
	\begin{figure}
		\label{Pic4}
		\center \includegraphics[width=3.5in,height=3.2in,clip,keepaspectratio]{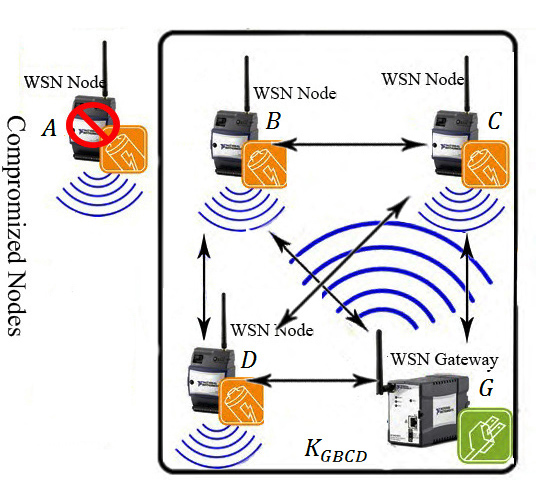}
		\caption{Adding a new node to the group}\color{black} 
	\end{figure}
	
	\textbf{Broadcast encryption.} A practical MP-NIKE scheme can be used to construct a broadcast encryption protocol \cite{4}. Broadcast encryption is a way for broadcasting an encrypted message on a public channel, such that broadcasting server can be sure that just a group of authorized users can decrypt the message. The size and members of the group of authorized users is not constant, and it may change for any single message. Previous broadcast encryption schemes suffer from security faults like ciphertext size and public key size \cite{Brod}. By using a practical MP-NIKE scheme, we can construct an efficient broadcast encryption scheme and solve the problem of long public key and ciphertexts \cite{4}.\\	
	
	\textbf{Key management for WSN.} One of the other applications of the MP-NIKE is key management in Wireless Sensor Networks. Since message exchanging consumes battery power and, on the other hand, sensors have energy limitations, using interactive key exchange schemes, such as \cite{purohit11}, \cite{IEEEAccess2} and \cite{Ismail15}, is not a proper solution to protect confidentiality in this sort of networks. On the one hand, using a single master key for all nodes has very low resilience, and if an adversary captures only one node, he will be able to compromise all nodes of the WSN. On the other hand, if we use a distinct pair-wise key for any two nodes of $n$ nodes, then any node must store $n-1$ different keys. This solution creates a heavy storage burden on each node, also for adding a new node in the future, we would need to update the key chain of every single node. 
	
	Other pair-wise key distribution solutions, such as closest pair-wise key pre-distribution \cite{Liu03}, random pair-wise key scheme \cite{Chan03}, pair-wise key establishment protocol \cite{Zhu03}, combinatorial design-based pair-wise key pre-distribution scheme \cite{Camtepe04}, etc. try to adjust the storage problem and yet provide key resilience, but none of them can guarantee key resilience with O(1) storage complexity. Furthermore, compromised nodes are a challenge in WSN, since they can behave arbitrarily and cooperate with others. Several solutions have been proposed in the literature to address this problem (e.g., \cite{compro} or \cite{compro2}). Unfortunately, these solutions are highly demanding in terms of message exchange. Note that pairing-based two-party NIKE schemes impose massive computation cost and need powerful hardware, which is not considered to be affordable in WSN nodes, because of power and price limitations. \color{black} The Simple Password-Based Encrypted Key Exchange (SPAKE2) protocol \cite{SPKAE2} provides forward secrecy, but, in addition to just being two-party, it requires agreement on a password between any two users and \color{black} it is interactive. \color{black} It means that any two users need to somehow agree on a shared password, so, it is not efficent at all.
	
	Conversely, the proposed MP-NIKE scheme provides resilience and enables any subset of sensors to efficiently compute a shared key without any interactions; also, adding new nodes in future can be quickly addressed. Besides, removing a compromised node will be quickly done. Since the proposed scheme is based on modular exponentiation, its computation cost is lower than pairing-based two-party NIKE schemes. The overall scenario of employing MP-NIKE in WSN  is shown in Fig. 1, Fig. 2 and Fig. 3. In Fig. 1, a group of four honest nodes $A, B, C$ and $G$ can easily and securely communicate by a shared key $K_{ABCG}$ which is computed by each of them without any interaction with the others. Fig. 2 shows how they can easily remove a compromised node $A$ just by using the shared key $K_{BCG}$ computable only by $B, C$ and $G$. The Fig. 3 illustrates that adding a new node can be quickly done.\\
	\subsection{Our Contribution}
	\label{Con}
	In 2014, Eskeland introduced a fully resilient and efficient MP-NIKE scheme based on modular exponentiation in RSA modulus \cite{2}. Nevertheless, in this paper, we propose a Euclidean algorithm-based attack against this scheme, that invalidates Eskeland's claim.
	
	Furthermore, we describe a new fully resilient MP-NIKE scheme. The computation cost of the proposed protocol for computing a shared key for a  group $W$, with size $|W|$, is only $|W|-1$ modular exponentiations and every user needs to store just one public/private key pair of small size. 
	
	\color{black} For the sake of comparison, obtaining a shared key among 19 users at 80-bit security level by using the 5Gen multilinear map (an extension of the CLT13) \cite{5Gen}, takes 33 seconds. However, in our scheme, each user requires to compute 18 modular exponentiations in a 1024-bit modulus so that it takes about 281 milliseconds. \color{black} We proved that the security of our scheme is equal to Fiat-Naor problem, (the security of Fiat-Naor scheme is based on the root extraction in RSA modulus, which is equal to RSA problem \cite{FN}). Also, since our proposal, compared with previous ones, is based on lighter cryptographic operations such as modular exponentiations, its computation cost is low.
	
	\subsection{Paper Organization}
	The rest of this paper is organised as follows. \color{black} Section \ref{Re} introduces the related work. \color{black} In Section \ref{Pre}, the general model of MP-NIKE, as well as the required preliminaries and definitions, are described. In Section \ref{Previous}, we briefly review the Fiat-Naor and the Eskeland schemes and propose an attack on the Eskeland scheme. Our novel MP-NIKE scheme is introduced in Section \ref{prop}. The security of the proposed scheme is discussed in Section \ref{Security}. \color{black}In Section \ref{performance} we compare computational and memory overhead of the proposed scheme against previous MP-NIKE schemes. \color{black}Finally, we conclude the paper in Section \ref{Conclu}.
	
	\section{Related Work}
	\label{Re}
	The first non-interactive key exchange scheme was introduced in the seminal work of Diffie and Hellman in 1976 \cite{DH76}. Their proposal was secure and efficient, but it was bounded only to the two-party case. In 2008, Cash et al.\ introduced a new assumption, called the Twin Diffie-Hellman problem, and presented a new two-party NIKE scheme that is secure in the random oracle model \cite{Cash}. In 2013, Freire et al.\ presented a new two-party NIKE scheme and proved its security using a game-based security model. However, they did not model the key registration process in their security model \cite{NIKE}. A year later, Freire et al.\ successfully modeled the key registration process in the security model of their new two-party NIKE scheme \cite{NIKE2}.
	
	The first three-party non-interactive key exchange scheme was introduced by Joux in 2001 using bilinear pairings \cite{J01}. His protocol was identity-based, secure and efficient, but it could not be extended for more than three parties. \color{black}Fiat and Naor first proposed the idea of MP-NIKE \color{black}in 1992, where they also suggested the employment of this idea in broadcast encryption. Unfortunately, their scheme is only 1-resilient \cite{FN}, i.e. the keys are protected only against one user; in other words, any two or more colluding members could compute the shared key of any group, whether they belong to this group or not. In 2003, Boneh and Silverberg showed that multilinear maps could be used to construct multi-party NIKE schemes \color{black} (see Definition \ref{mmapd} and Lemma \ref{rmmap}), but they also revealed that the bilinear Tate and Weil pairings could not be generalized to multilinear maps \cite{4}. So, they could not propose any concrete multilinear map \cite{4}. \\
	The design of a full-resilient, multi-party and non-interactive, key-exchange protocol remained an open problem until Garg, Gentry and Halevi \cite{5} introduced the first multilinear map (GGH13) based on lattices. After that, Langlois in 2014 presented a more efficient GGH map called GGHLite \cite{10}, and then, Albrecht et al.\ proposed the first practical MP-NIKE scheme based on GGHLite \cite{11}. This scheme, however, is inefficient in public parameters size and computational cost, as they declared in 80-bit security, computing a shared key between 7 users will take about 1.75 seconds on a 16-core CPU. In 2015, Hu and Jia showed that GGH and GGHLite are insecure \cite{13}. In 2015, Gentry et al.\ proposed a graph-induced multilinear map from lattices \cite{12}, which for simplicity we call it GGH15. Recently, Coron et al.\ showed that GGH15 is insecure too \cite{14}. Moreover, all of these multilinear maps are inefficient because they are based on using cryptographic operations with a high computational cost.
	Coron et al.\ in 2013, proposed another candidate construction of multilinear maps over integers \cite{15}, denoted by CLT13 for simplicity, which soon was broken by Cheon et al.\ \cite{16}. Then, Coron et al.\ fixed their scheme \cite{CLT15}, but this time Minaud and Fouque proposed an attack on this fixed scheme and downgraded it to the previous one \cite{17}.  
	Recently, Ma and Zhandry proposed another multilinear map based on CLT13 which is provably secure against previously known attacks, but its security its not proven in the standard security model \cite{MZ17}. Their scheme is a modified version of CLT13, so it is not efficient.\\
	
	Multilinear maps and MP-NIKE schemes can be constructed by using $i\mathcal{O}$ \cite{8}. Garg et al.\ proposed the first construction of $i\mathcal{O}$ for general boolean circuits by using multilinear maps \cite{6}, and then Rao \cite{3}, Yamakawa et al.\ \cite{7}, Boneh and Zhandry \cite{8} and Khurana et al.\cite{9} proposed some multilinear maps and MP-NIKE schemes based on $i\mathcal{O}$ and constrained Pseudo-Random Generator (PRG) \cite{9}. Since multilinear maps are needed to construct $i\mathcal{O}$ \cite{5}, creating a multilinear map using $i\mathcal{O}$ seems impractical. To the best of our knowledge, so far no provably secure and practical $i\mathcal{O}$ and multilinear maps have been proposed in the literature. So, all of the $i\mathcal{O}$-based and multilinear map-based MP-NIKE schemes are either impractical or insecure \cite{8,13,14,16,17}. \color{black} 
	
	In 2016, Chen et. al. proposed an identity-based MP-NIKE based on Witness Pseudo Random Function(WPRF) \cite{SOK}. Constructing a WPRF needs asymmetric cryptographic multilinear map \cite{Z14}. If we had access to an efficient multilinear map, we would use it to construct an MP-NIKE from the scratch \cite{4}.\\
	
	MP-NIKE schemes are also called Non-Interactive Conference Key Distribution System (NICKDS) schemes. In 1998 Blundo et al. \cite{Blundo98} proposed a $k$-secure $t$-conference NICKDS using multi-variable symmetric polynomials, which was secure against $k$ colluding users. In Blundo scheme, no $k$ colluding users do not obtain any information about any key of other users, but any $k+1$ colluding users can obtain key of all other users \cite{Blundo98}. The bound of security parameter $k$ in Blundo scheme with $n$ users is equal to $n-t$, where $t$ is the size of the conference, and each user needs to store a piece of information with the size of $\binom{k+t-1}{t-1}$ times the size of common key \cite{Blundo98}. It means that for calculating a 128-bit key, in 50-secure 50-conference NICKDS system (for 100 users), each user needs to store about $2^{103}$-bit data, which is infeasible.\\
	
	\section{Preliminaries}
	\label{Pre}
	In this section, we introduce the notations used in this paper, the necessary definitions, lemmas and the general model of Non-Interactive Key Exchange (NIKE).
	\color{black}\subsection {Notations}
	The Notations and abbreviations used in this paper are outlined in Table 1.
	\begin{table}
		\label{Notation}
		\caption{ \color{black}  Notations and abbreviations}
		\begin{center}
			%35 55 40 55
			\begin{tabular}{|p{40pt}|p{140pt}|}
				\hline
				\color{black} \textbf{Symbol} & \color{black}  \textbf{Description}    \\
				\hline
				 $ p$, $ q$, $z$ & Three large primes\\
				\hline
				$N$ & Modulus of computation\\
				\hline
				$\mathds{G}$ & A multiplicative subgroup of $\mathds{Z}^*_N$\\
				\hline
				$g$ & A generator of $\mathds{G}$\\
				\hline
				\textsf{O}$()$ & Random Oracle\\
				\hline
				$H()$ & A one-way hash function\\
				\hline
				$W$ & A group of users aiming to compute a shared key \\
				\hline
				$K_W$ & Shared key of the group $W$\\
				\hline
				$U_i$ & User $i$\\
				\hline
				$PP$ & Public parameters\\
				\hline
				$msk$ & The KGC's master secret key\\
				\hline
				$e_i$ & Public key of user $i$\\
				\hline
				$d_i$ & Private key of user $i$\\
				\hline
				$\gamma$ & Security parameter\\
				\hline
				$A$ & The adversary\\
				\hline
				$B$ & The challenger algorithm\\
				\hline
				$W_c$ & Challenge set\\
				\hline
				$s$ & $|W_c|$\\			
				\hline
			\end{tabular} \\\setlength\tabcolsep{5pt}

		\end{center}
	\end{table}

	\color{black}
	\subsection {NIKE General Model}
	\label{Gene}
	The general model of MP-NIKE consists of the following four algorithms.

	\begin{enumerate}
		\item \textit {\textbf{Setup}($\gamma $).} Given a security parameter $\gamma $ as input, the Key Generation Center (KGC) runs setup($\gamma $) to generate a master secret key $msk$ and a set of public parameters $PP$. Then, the KGC announces $PP$ to all users.
		
		\item \textit{ \textbf {KeyGen}($PP$, $msk$, $i$).} In this phase, the user $U_i$ is authenticated by the KGC, and the KGC generates a valid public/private key pair ($e_i$, $d_i$), saves them in its database and finally gives them to user $U_i$.
		\item \textit { \textbf {SharedKey}($PP$, $d_i$, $\{e_j: U_j \in W$, $j \neq i\} $).} Each member $U_i$ of group $W$ runs this algorithm to obtain a valid shared key $K_W$. All members of $W$ are able to compute this shared key, while other users which do not belong to $W$ cannot compute it.
		\item \textit{ \textbf{Join}($PP$, $K_W$, $e_s$).} Suppose there exists a group $W$ with a key $K_W$ shared among its members. When a new user $U_s$ wants to join this group to form a new group $W\acute{}=W \cup \{U_s\}$, it must run \textit { \textbf{SharedKey}$(PP$, $d_s$, $\{e_j: U_j \in W\acute{}$, $j \neq s \})$}, while the other members of $W$ can update the shared key $K_W$ to $K_{W\acute{}}$ with a lower computational  overhead. The output of this algorithm is a new shared key $K_{W\acute{}}$, where $W\acute{}=W \cup \{U_s\}$.
	\end{enumerate}
	\subsection {Definitions and Lemmas}
	To prove the security of our scheme, we need the following definitions and lemmas. 
	%1 
	%\begin{definition}(Safe modulus): \textit{Let $N=p\acute{} \cdot q\acute{}$ be the product of two large primes $p\acute{}$ and $q\acute{}$. Then, $N$ is called a safe RSA modulus if $p\acute{}=2pz+1$ and $q\acute{}=2q+1$, where $p$ and $q$ are also prime numbers.}
	%\end{definition}
	%1
	\begin{lemma}
		\label{SafeRSA} (Generators of a subgroup): \textit{Let $N=p\acute{}q\acute{}=(2pz+1)(2q+1)$ be a RSA modulus. Besides, assume that  $\mathds{G}$ is a subgroup of order $pzq$ in $\mathds{Z}_N^*$ and $\hat{\mathds{G}}$ is a subgroup of order $zq$ from $\mathds{G}$. Under these conditions, if $g$ is a generator of $\mathds{G}$ then $g_1=g^{p} \bmod N$ is a generator of $\hat{\mathds{G}}$ \cite{HB}.}
	\end{lemma}
	%2
	\begin{definition} ($q$-th residue): \textit{Let $N=p\acute{} q\acute{}=(2pz+1)(2q+1)$ be a RSA modulus and $\mathds{G}$ be a subgroup of order $p z q$ in $\mathds{Z}_N^*$.  Then, an element $g_1$ is a $zq$-th residue in $\mathds{Z}^*_N$, if there exists at least one element $\alpha \in \mathds{G}$ such that ${\alpha}^{zq} \bmod N=g_1$.}
	\end{definition}
	%3
	
	\begin{definition} (The Fiat-Naor problem): 
		\label{FNP}	\textit{
			Let $N=\acute{p}\acute{q}=(2pz+1)(2q+1)$ be a RSA modulus. Besides, assume that $g$ is a private generator of the subgroup $\mathds{G}$ of $\mathds{Z}^{*}_{N}$ of order $q$. Under these conditions, and given $(y, g^y \bmod N, c)$, where $c$ is coprime to $y$, compute $g^c \bmod N$ \cite{FN}. It is assumed that the Fiat-Naor problem is intractable and equivalent to the root extraction in RSA modulus and RSA problem \cite{FN}. }
	\end{definition}
	\color{black}
	\begin{definition} (Multilinear map (mmap)): 
		\label{mmapd}	\textit{
			Let $\mathds{G}_1$ and $\mathds{G}_2$ be two multiplicative groups of the same prime order. The map $e:\mathds{G}_{1}^{n}\rightarrow \mathds{G}_{2}$ is an $n$-multilinear map, if it satisfies the following two properties \cite{4}:
			\begin{enumerate}
				\item Multilinearity. If $a_1,\dots, a_n \in \mathds{Z}$ and $x_1,\dots, x_n \in \mathds{G}_1$ then\\
				\begin{equation}
				e({x_1}^{a_1},\dots, {x_n}^{a_n})=e({x_1},\dots, {x_n})^{a_1\dots a_n}
				\end{equation}
				\item Non-degeneracy. Let $g \in \mathds{G}_1$ be a generator of $\mathds{G}_1$, then $e(g,\dots, g)$ must be a generator of $\mathds{G}_2$. 
			\end{enumerate} 
		}
	\end{definition}
	\begin{lemma}(Realization of MP-NIKE by using multilinear map (mmap)): 
		\label{rmmap}	\textit{An MP-NIKE scheme can be performed by using multilinear maps. Let the map $e:\mathds{G}_{1}^{n}\rightarrow \mathds{G}_{2}$ be an $n$-multilinear map, $g$ represent a generator of $\mathds{G}_1$ and $g_t=e(g,\dots ,g)$ be a generator of $\mathds{G}_2$. Suppose that members of group $W=\{U_i,\dots ,U_{n+1}\}$ with $n+1$ users want to compute a shared key $K_W$. Any user $U_i \in W$ chooses a random integer $a_i$ and publishes $g^{a_i}$ as its public key. Now any user $U_i \in W$ can compute the shared key $K_{W}$ as below.
			\begin{eqnarray}
			K_W=e({g}^{a_1},\dots, {g}^{a_{i-1}}, {g}^{a_{i+1}},\dots, {g}^{a_{n+1}})^{{a_i}}=g_{t}^{a_1\dots a_{n+1}}			
			\end{eqnarray}}
	\end{lemma}
	\color{black}
	\subsection {MP-NIKE Security Model}
	\label{SecModel}
	We adopt the security model of Non-Interactive Conference Key Distribution (NICKD) of \cite{SJ8}, which is an extension of the Bellare-Rogaway security model \cite{Bellare} and \cite{Bresson}. In this security model, the adversary is allowed to use three types of oracles: \textsf{Test}, \textsf{Reveal} and \textsf{Corrupt}. The adversary can adaptively corrupt users of his choice and obtain corrupted users' keys by using \textsf{Corrupt} oracle. By using the \textsf{Reveal} oracle, the adversary can obtain the shared key of any arbitrary group \color{black} as below. The adversary gives an arbitrary group $W$ to the \textsf{Reveal} oracle as input, and then this oracle sends the repective shared key as output to the adversary. \color{black} 
	A \textsf{Corrupt} query can also obtain the information leaked by a \textsf{Reveal} query \cite{SJ8}, since the adversary can compute the shared key of any arbitrary group $W$ just by corrupting one user $U_i \in W$. So, we omit the \textsf{Reveal} oracle in our model.\\  
	
	To prove the security of an MP-NIKE scheme, by contradiction, we suppose there exists an adversary $A$ that can distinguish between a random bit string and the shared key of a specific group $W^*$ of arbitrary size $s$, which it does not access to private keys of its members. Then, we show that there exists an algorithm $B$ that can solve a hard problem by invoking algorithm $A$. In this model, at the first step, the adversary $A$ must commit to $s$, by announcing it to algorithm $B$. 
	%Suppose that there is an algorithm $C$ which gives the public parameters $PP$ and a public/private key pair to the $B$. The algorithm $B$ gives $s$ to the algorithm $C$ and then the algorithm $C$ gives a challenge set $W_c$ of size $s$ to the algorithm $B$.
	Now the algorithm $B$ as a challenger plays the following game with adversary $A$.\\
	For the sake of formalisation, we adapt the model of \cite{SJ8} to five phases as below. %
	\begin{enumerate}
		\item \textit{\textbf{Commit.}} At the first step, the adversary A must choose $s$ and then commit to $s$ by sending it to the algorithm $B$.
		
		\item \textit{\textbf{KeyGen.}}
		The algorithm $B$ generates $q_c$ sets $\langle W_1,\dots, W_{q_c} \rangle$, where each of them contains exactly $s$ valid public keys. It also generates the associated private keys and keeps them private. Finally, it gives these $q_c$ sets $\langle W_1,\dots, W_{q_c} \rangle$ in a random order to the adversary.
		
		\item \textit{\textbf{Phase 1.}} In this phase, the adversary is allowed to ask $q_c$ \textsf{Corrupt} queries and one \textsf{Test} query from the algorithm $B$. The formal definition of these oracles is as follows.
		\begin{itemize}
			\item \textsf{\textit{\textbf{Corrupt}($U_i$)}} \color{black} (or equivalently \textsf{\textit{\textbf{Corrupt}($e_i$)}}). \color{black} The adversary can corrupt any user $U_i$ adaptively \color{black}by using \textsf{Corrupt} oracle. The adversary gives $U_i$ as input to \textsf{Corrupt} oracle and then this oracle sends the respective private key $d_i$ as output to the adversary. 
			\color{black}
			\item \textsf{\textit{\textbf{Test}($W^*$).}} When the adversary decides to terminate Phase 1, it chooses one of the $q_c$ input sets of public keys, say $W^*$, such that for all $U_i \in W^*$, $U_i$ should not have appeared in none of the \textsf{Corrupt}($U_i$) queries, and then sends \textsf{Test}($W^*$) to the challenger. After receiving this query, the challenger generates a random bit $b \in \{ 0, 1\}$: if $b=0$ then the challenger sends $K_{W^*}$ to the adversary, otherwise, it generates a random string $\mathit{rand} \leftarrow \{ 0, 1\}^{\lambda}$ and sends it back to the adversary, where $\lambda$ is the bit length of $K_{W^*}$. 
		\end{itemize}
		
		\item \textit{\textbf{Phase 2.}} This phase is the same as Phase 1, except that the adversary does not have access to the \textsf{Test} oracle and it is not allowed to ask  \textsf{Corrupt}($U_i$), where $U_i \in W^*$. 
		
		\item \textit{\textbf{Guess.}} In this phase, the adversary guesses $b$ by a bit $b\acute{} \in \{0, 1\}$ and sends it to the challenger. If $b\acute{} = b$ the adversary wins the game.
	\end{enumerate}

	Finally, we give the following definition concerning the security offered by the MP-NIKE scheme.  
	
	\begin{definition} (Fully resilient MP-NIKE scheme): \textit{
			\label{ep}
			MP-NIKE scheme $\mathcal{E}$ is fully resilient ($q_c, T, \epsilon$)-secure, if in the above described game any adversary $A$, which is allowed to ask $q_c $ queries from \textsf{Corrupt} oracle, cannot distinguish a random string (i.e., $\mathit{rand}$ from the true shared key $K_{W^*}$) with an advantage greater than $\epsilon$ in time $T$.} 
		\begin{equation}
		Adv_A^\mathcal{E}=\Big| Pr[b\acute{}=b] - \frac{1}{2} \Big| \leq \epsilon .
		\end{equation}
	\end{definition}
	
	\section{The Fiat-Naor and the Eskeland NIKE schemes}
	\label{Previous}
	In this section, we briefly introduce the Eskeland NIKE scheme, but first we need to review the Fiat-Naor scheme. Then, we present a coalition attack against the Eskeland's NIKE scheme in which any two colluding users can compute the shared key of any other group of users.
	
	\subsection{The Fiat-Naor Scheme}
	This scheme is based on the intractability of the factorisation of RSA moduli and is secure against any adversary that has access to at most one public/private key pair. This scheme works as follows.
	
	The KGC generates an RSA modulus $N=pq$, where $p$ and $q$ are large primes, and then selects at random a generator $g$ from $\mathds{Z}_N^*$ and keeps it as a master secret key $msk$ for itself. To generate a valid public/private key for user $U_i$ the KGC selects at random a prime $y_i$ and then computes the private key of the user $U_i$ as $d_i=g^{y_i} \bmod N$. Finally, the KGC gives $e_i=y_i$ and $d_i$ as public key and private key to user $U_i$. Let $W$ be a group of users aiming to compute a shared secret key. User $U_i$, where $U_i \in W$, computes the shared key as follows.
	\begin{equation}
	K_W=d_i^{\prod\limits_{j: U_j \in W, j \neq i}e_j} \bmod N = g^{\prod\limits_{j: U_j \in W}y_j} \bmod N.
	\end{equation}
	Suppose that a new user $U_s$ wants to join $W$ to form a new group $W\acute{}=W \cup \{U_s\}$. It must run \textit{\textbf{SharedKey}}($PP$, $d_s$, $\{e_j: U_j \in W\acute{}, j \neq s \} $) to compute a new shared key as below.
	
	\begin{equation}
	K_{W\acute{}}=d_s^{\prod\limits_{j: U_j \in W}e_j}.
	\end{equation}
	Any other user $U_j$ who has already joined $W$ simply computes
	\begin{equation}
	K_{W\acute{}}={K_W}^{e_s} \bmod N.
	\end{equation} 
	
	The Fiat-Naor NIKE scheme is 1-resilient but insecure against collaboration of two or more adversaries. In other words, any adversary accessing at least two public/private key pairs can compute $g$ by using the Euclidean algorithm and can break the Fiat-Naor scheme \cite{FN}. Stated differently, given ($a, g^a \bmod N$) and ($b$, $g^b \bmod N$), the adversary can compute $g^{gcd(a, b)} \bmod N$, by using the Euclidean algorithm and performing a sequence of modular exponentiations on $g^b \bmod N$ and $g^a \bmod N$. Note that $a$ and $b$ are prime, so the result equals $g$. 
	
	\subsection{The Eskeland Scheme}
	Let $N=pq$ be an RSA modulus, $g$ be a public generator of $\mathds{Z}_N^*$ and $H()$ be a secure one-way hash function. The KGC selects at random a secret $u$. The public key of user $U_i$ is computed as $e_i=H($identity of user $U_i$) and its private key is generated by the KGC as $d_i=z_i u+v_i \varphi (N)$, where $z_i=e_i \bmod \varphi (N)$ and $v_i$ is a unique random element of $\mathds{Z}_N^*$. Then, any user $U_i \in W$ computes the pre-shared key of group $W$ as below.
	\begin{equation}
	K_W=g^{d_i \prod\limits_{j: U_j \in W, j \neq i}e_j} \bmod N=g^{u\prod\limits_{j: U_j \in W}z_j} \bmod N.
	\end{equation}
	\subsection{Attack on the Eskeland Scheme}
	\label{attack}
	Eskeland claimed that his scheme is fully resilient against  any number of colluding adversaries \cite{2}. The colluding adversaries have access to each other's private keys and also the shared key of the groups to which they belong. Nevertheless, they do not have access to the master secret key of the KGC, nor to the private key of honest users (i.e. secrecy of private keys -- Security Requirement-1 in \cite{2}). However, we perform an attack by removing the effect of multipliers of $\varphi(N)$ and show how the secrecy of groups keys (Security Requirement-2 in \cite{2}) is not guaranteed. Mathematically, we show this below: \\

	Let $e_i$ and $e_j$ be two public keys in the Eskeland scheme and $gcd(e_i, e_j)=1$. By using the extended Euclidean algorithm we can efficiently compute the integers $a$ and $b$ such that $ae_i-be_j=1$ \cite{19}.  
	
	\begin{lemma}
		\label{ESK}
		Let $e_i$ and $e_j$ be two (coprime) public keys in the Eskeland scheme. Besides, we assume that $a$ and $b$ are two integers such that $ae_i-be_j=1$. Then, it is satisfied that $(az_i-bz_j) \bmod \varphi (N)=1$. 
	\end{lemma}
	
	\begin{proof}
		$e_i=z_i+k_i \varphi (N)$ for some integer $k_i$, and $e_j=z_j+k_j \varphi (N)$ for some integer $k_j$. If $ae_i-be_j=1$, then we have that 
		\begin{align}
		ae_i-be_j&=(az_i-bz_j)+(ak_i-bk_j)\varphi(N) \nonumber\\
		&=(az_i-bz_j) \bmod \varphi (N)=1.
		\end{align} 
	\end{proof}
	
	Based on Lemma \ref{ESK}, if $ae_i-be_j=1$, any two colluding users $U_i$ and $U_j$ can compute $\acute{u}=ad_i-bd_j$ which is equal to $u \bmod \varphi(N)$ as below.
	\begin{align}\begin{split}
	u\acute{}=ad_i-bd_j
	&=(az_i-bz_j)u+(av_i-bv_j) \varphi (N)\\
	&=(az_i-bz_j)u \bmod \varphi (N)=u \bmod \varphi(N) .
	\end{split}
	\end{align} 
	Then, given $\acute{u}$, the colluding users $U_i$ and $U_j$ can compute $g^u=g^{u\acute{}} \bmod N$. They can subsequently compute $g^{u\prod\limits_{U_i \in W\acute{}}e_i} \bmod N$ as the shared key of any arbitrary group $W\acute{}$.

	\section{Novel MP-NIKE Scheme}
	\label{prop}
	In this section, we present a new MP-NIKE scheme, which is secure against any coalition of malicious users. %The main idea of this scheme is inhibiting Euclidean algorithm by blinding $y_i$'s.
	
	The proposed scheme consists of four phases as below.
	
	\begin{enumerate}
		\item \textit{ \textbf{Setup($\gamma $).}} For a given security parameter $\gamma$, the KGC produces a RSA modulus $N=p\acute{} q\acute{}=(2pz+1)(2q+1)$, where $p$, $q$, $z$, $p\acute{}$ and $q\acute{}$ are large primes. Let $\mathds{G}$ be a subgroup of $\mathds{Z}_N^*$ of order $pzq$ and $m=\lceil \log_2 pzq \rceil $ be the bit length of $pzq$. The KGC selects at random a generator $g \in_R \mathds{G}$ and presents a hash function $H: \mathds{Z}^*_N\longrightarrow \{0, 1\}^{\lambda}$, i.e. the output of $H(\cdot)$ is a $\lambda$-bit string.
		
		Finally, the KGC publishes $\{ N, H(\cdot), g^p\}$ as the public parameters $PP$ and keeps the master secret key $msk=(p,z, q)$ for itself.
		\item \textit{ \textbf{KeyGen($PP$, $msk$).}} The KGC selects two uniformly random $m/2$-bit odd integers $y_i$ and $k_i$, and computes $e_i=(p y_i+zqk_i), d_i=g^{py_i} \bmod N$ and sends $(e_i, d_i)$ as a valid public/private key pair to user $U_i$.
		\item \textit{ \textbf{SharedKey($PP$, $d_i$, $\{e_j: U_j \in W$, $j \neq i\} $).}} Suppose $W$ is a subset of users of size $|W|$, aiming to obtain their preshared key $K_W$. Given the public key of other users of $W$, each user $U_i \in W$ computes $K_W$ as below.
		
		\begin{align}
		F_W&={d_i}^{\prod\limits_{j: U_j \in W, j \neq i}e_j} = {g^{\Big( p^{|W|} \prod\limits_{j: U_j \in W}y_j \Big)} \bmod N },\\
		K_W&=H(F_W).
		\end{align}
		
		\item \textit{ \textbf{Join($PP$, $F_W$, $e_s$).}} Suppose a user $U_s$ wants to join a group $W$ to extend it to $W\acute{}=W \cup \{U_s\}$. It must run \textit{\textbf{SharedKey}($PP$, $d_s$, $\{e_j: U_j \in W \} $)}, while other users in $W$ can simply compute the shared key of group $W\acute{}$ as below.
		\begin{align} 
		F_{W\acute{}}&=F_W^{e_s} \bmod N,\\
		K_{W\acute{}}&=H(F_{W\acute{}}).
		\end{align}
	\end{enumerate}
	
	Note that if $(e_i, d_i)$ and $(e_j, d_j)$ are valid public/private key pairs, then for any pair of integers $\alpha $ and $\beta $, $(\alpha e_i+\beta e_j, d_i^\alpha d_j^\beta  \bmod N)$ are also valid public/private key pairs.
	In the proposed scheme, like many other cryptographic schemes (e.g. Boneh and Franklin's identity-based encryption scheme \cite{BF}), the adversary can generate some random public/private key pairs from other valid public/private key pairs, but it cannot find the private key for a given public key.

	\subsection{The Proposed MP-NIKE as a Broadcast Encryption scheme}
	\label{app}
	In this section, we propose a broadcast encryption scheme based on the proposed MP-NIKE protocol. The proposed MP-NIKE protocol can be used to construct practical broadcasting encryption with low computation cost and low message overhead. The proposed scheme consists of three phases as below \cite{Brod}.
	
	\textit{ \textbf {Brod\_Setup}($ \eta, \gamma $).}
	The broadcasting server runs \textit{Setup($\gamma $)} to obtain the public parameters $PP$ and $msk$. Then, it generates $\eta$ public/private key pairs by runing \textit{KeyGen($PP$, $msk$)} procedure and saves them in its database. Finally the broadcasting server gives public/private key    pairs to $\eta$ users.\\
	
	\textit{\textbf{Brod\_Encrypt}($W, PP, M$).}
	Suppose the broadcasting server wants to encrypt a message $M$ for a group $W$ of authorised users. At the first step, it obtains a private key $d_i$ of one user $U_i$ where $U_i \in W$ from its database. Now the broadcasting server runs  \textit{SharedKey($PP$, $d_i$, $\{e_j: U_j \in W$, $j \neq i\} $)}    to obtain $K_{W}$. Then it encrypts a message $M$, for example by the AES algorithm, under the group key $K_{W}$ and sends the resulted ciphertext $CT = Enc_{K_W}(M)$ along with the list of public keys of authorised users $W$ over the broadcasting channel.
	
	\textit{\textbf{Brod\_Decrypt}($W, i, d_i, PP, CT$)}. 
	Any user $U_i$ from the authorized group $W$ computes \textit{SharedKey($PP$, $d_i$, $\{e_j: U_j \in W$, $j \neq i\} $)} and decrypts the ciphertext $CT$ by using $K_{W}$.\\
	
	In addition to broadcast encryption, the proposed scheme can be used in various applications such as IoT or WSN, where a set of smart devices (e.g. in the home or deployed in the countryside)  want to communicate securely.  In this scenario, the IoT devices may have different access permissions, and some messages must not be decryptable to some elements. Besides, when a node is compromised, it must not be able to decrypt any message. In this case, the other nodes of the network can compute a shared key for a group in which the compromised node is not a member.  
	
	\section{Security Analysis}
	\label{Security}
	In this section, we prove the security of our proposal in two steps. First, we show that our scheme is 1-resilient. Second, we demonstrate that if the proposed MP-NIKE scheme is 1-resilient then it is also full-resilient. This proof uses the random oracle model based on the adopted security model of Non-Interactive Conference Key Distribution (NICKD) of \cite{SJ8}, which is an extension to the Bellare-Rogaway security model \cite{Bellare} and \cite{Bresson}.

	\subsection{Step 1: The proposed scheme is 1-resilient}
	Theorem \ref{Th1} shows that the proposed scheme is 1-resilient and proves that given $PP$ and one public/private key pair, the adversary is not able to obtain private key of any other user.     
	\begin{theorem} %1
		\label{Th1}
		(The proposed scheme is 1-resilient): Let $N=p\acute{} q\acute{}=(2q+1)(2zp+1)$ be a RSA modulus. In addition, assume that $\mathds{G}$ is a subgroup of $\mathds{Z}^{*}_{N}$ of order $p zq$, and $g$ is a generator of $\mathds{G}$. Also, $\mathds{G_1}$ is a subgroup of $\mathds{Z}^{*}_{N}$ of order $zq$ and $g^p$ is a generator of $\mathds{G_1}$. Under these assumptions, if the Fiat-Naor problem is intractable, then the proposed scheme is 1-resilient. 
	\end{theorem}
	\begin{proof}
		Suppose there is an adversary $A$ that given the public parameters $PP$ and one public/private key pair ($e_{\alpha }=y_{\alpha} p+k_{\alpha }zq$, $d_{\alpha}=g^{py_{\alpha} } \bmod N$), can obtain the private key of a challenge public key $e_c=y_cp+k_czq \in \mathds{Z}_N^*$, which is equal to $d_{c}=g^{py_c} \bmod N$. Then, we show that there is an algorithm $B$ that can solve an instance of the Fiat-Naor problem over the group $\mathds{G_1}$ with generator $g^p \bmod N$.
		
		\textit{\textbf{Phase $1$.}}  Algorithm $B$ is given the modulus of computation $N$ as well as an instance of the Fiat-Naor problem $(y_\alpha$, $(g^p)^{y_\alpha})$ and $e_r=p+k_rzq$ as input, \color{black} and is asked to output the Fiat-Naor private key for a specific public key $y_c$ (which is supposed to be $g^{py_c} \bmod N$). Note that if someone choose a safe RSA modulus ($N=(2p+1)(2q+1))$ the $e_r$ can be easily computed as $e_r=(N-1)/2=2pq+p+q=p+(2p+1)q$ so giving it as input to algorithm $B$ does not give any information. 
		
		Algorithm $B$ first transforms the given Fiat-Naor key pair to a NIKE key pair. To do this, it computes $e_{\alpha}=e_ry_\alpha=p(y_\alpha)+zq(k_ry_\alpha)$. Note that $g^{py_\alpha} \bmod N$ is the corresponding NIKE private key for the public key $e_\alpha$.
		Moreover, algorithm $B$ changes the given challenge $y_c$ to a valid public key of the NIKE scheme $e_{c}=e_ry_{c}=p(y_{c})+zq(y_{c}(k_r))$. Finally it sends $e_{\alpha}, g^{py_\alpha}$ and $e_{c}$ as a valid key pair and a challenge public key in the proposed NIKE scheme to the adversary $A$.\\	
		\textit{\textbf{Phase $2$.}} Suppose the adversary $A$ is can compute the private key $d_{c}=g^{py_{c}} \bmod N$ and sends it back to $B$. It allows $B$ to solve the given instance of the Fiat-Naor problem. In detail, $B$ outputs $g^{py_{c}} \bmod N$, which is the answer to the given Fiat-Naor problem.	\color{black}
	\end{proof} 
	
	%\begin{cor} 
	%   In the proposed NIKE scheme, given a public/private key pair $(e_{\alpha}=y_{\alpha}p+k_{\alpha}q, d_{\alpha}=g^{py_{\alpha}} \bmod N)$ the adversary cannot compute the secret $y_{c_{i}}$ corresponding to challenged public key $e_{c_{i}}=y_{c_{i}}p+k_{c_{i}}q$
	%\end{cor}
	%\begin{proof}
	%   We prove the Lemma by contradiction. Suppose the adversary can compute $y_{c_{i}}$ for some public key $e_{c_{i}} = y_{c_{i}}p + k_{c_{i}}q$. Then, since $g^p$ is known to the adversary, s/he can compute $d_{c_{i}}={g^p}^{y_{c_{i}}} \bmod N$. The above contradicts the 1-resiliency of the proposed NIKE scheme proved in Theorem 1.    
	%\end{proof}
	In Theorem \ref{Th2}, we prove that in the proposed scheme an adversary that has access to several public/private key pairs does not have any advantage over the adversary that has access to only one public/private key pair, public parameters and the generator $g^p$. Note that since the adversary can compute $g^p \bmod N$, including $g^p \bmod N$ in the public parameters $PP$ does not cause any security flaw to the scheme.
	\subsection{Step 2: The proposed scheme is full-resilient}
	In Theorem \ref{Th2}, we prove that the proposed scheme is full-resilient in the random oracle model. In other words, it formally shows that in the proposed solution no coalition of adversaries can obtain other users' private keys or compute the shared key of some group $W\acute{}$ of which they are not valid members.
	Note that if the adversary can compute the private key of any member of the group, it will be able to calculate the shared key of the group too. Thus, Theorem \ref{Th2} also proves that the adversary is not able to compute the private key for another specific public key.
	
	\begin{theorem}
		\label{Th2}
		In the proposed MP-NIKE scheme, suppose that the hash function $H$ is modelled as a random oracle \textsf{O}. Let $A$ be an adversary who is allowed to ask $q_c$ queries from the \textsf{Corrupt} oracle. Besides s/he has advantage $\epsilon$ to distinguish a random string from the shared key of a group $W_c$ of size $s$ and has no access to the private keys of the group members. Then there is an algorithm $B$ that has at least an advantage of $\frac{\epsilon}{ e(q_c+1)}$ against the proposed scheme given only one public/private key pair.
	\end{theorem}
	
	\begin{proof}%2
		Suppose there is an adversary $A$ that given the public parameters $PP$ and several public/private key pairs can break the proposed scheme. Then we construct an algorithm $B$ that can break the security of the proposed scheme given only one public/private key pair (which contradicts Theorem 1). Suppose there is an algorithm $C$ which gives the public parameters $PP$ and a public/private key pair to algorithm $B$ as input and asks this algorithm to compute the shared key for a specific group $W_c$. Then, the algorithm $B$ as the challenger for $A$ gives $PP$ to $A$ and tries to respond to \textsf{Corrupt} queries issued by $A$. \color{black} Then the algorithm $B$ tries \color{black} to find the shared key for the group $W_c$ from the messages of adversary $A$. \color{black} In other words, the algorithm $B$ is challenged by a set $W_c$ of $s$ public keys and has to output the shared key $K_{W_c}$ as the response to this challenge (Note that $K_{W_c}$ is the shared key which the owners of these public keys can compute in the proposed MP-NIKE scheme). Algorithm $B$ plays the next game with adversary $A$ who can calculate the shared key of some group $W^*$ if it is allowed to receive $q_c$ private keys corresponding to $q_c$ public keys chosen by itself. In this game algorithm $B$, with some probability, responds to the Corrupt queries issued by $A$ and attempting to provide the $q_c$ private keys requested by $A$. Moreover, adversary $A$ is forced to obtain its required hash values only through sending requests to the random oracle which is controlled by algorithm $B$. Thus, by storing and using the values that  the adversary $A$ has asked for their hashes, the algorithm $B$ would find the key corresponding to group $W_c$.

		\color{black} For the sake of simplicity, we suppose that before starting the game the adversary $A$ must commit to the size of the group $W^*$, which it wants to attack. For computing the shared key for this specific group, the algorithm $B$ as the challenger plays the following game with the adversary $A$: 
		\begin{enumerate}
			\item \textit{\textbf{Commit.}} At the first step, the adversary chooses $s$, which is the size of the group $W^*$, and then commits to $s$.\\ 
			\item \textit{\textbf{Initialization.}} Let $N=p\acute{} q\acute{}=(2q+1)(2pz+1)$ be a RSA modulus, $\mathds{G}$ be a subgroup of $\mathds{Z}_N^*$ of order $pzq$. The algorithm $B$ sends $s$ to the algorithm $C$ and then $C$ gives the public parameters $PP = \{N, g^p \bmod N, H\}$, a challenge set $W_c=\{  e_{c_1}, e_{c_2},\ldots, e_{c_s} \}$, where  $e_{c_i}=y_{c_i} p+k_{c_i} zq$ for $i = 1, 2,\ldots, s$ are valid public keys, as well as a public/private key pair ($e_{\alpha}=y_\alpha p+k_\alpha zq, d_{\alpha}=g^{py_\alpha} \bmod N$) and $e_r$ to the algorithm $B$ as input. Note that if someone choose a safe RSA modulus ($N=(2p+1)(2q+1))$ the $e_r$ can be easily computed as $e_r=(N-1)/2=2pq+p+q=p+(2p+1)q$ so giving it as input to algorithm $B$ does not make security problem. 
			
			%Let $e_r=(N-1)/2=2pzq+p+q=p+(2pz+1)q=p+k_rq$, so
			The $(e_r,  g^p \bmod N)$ is a valid public/private key pair. Finally, the algorithm $B$ sends the public parameters $PP$ to the adversary $A$. 
			
			Throughout the next phase, we use a random oracle \textsf{O}$(\cdot)$ to simulate the hash function $H$. The challenger controls this random oracle, and the adversary can obtain the hash value for any arbitrary string by asking it from this random oracle. In other words, the adversary $A$ gives $F_W$ to the random oracle and then the random oracle generates a unique random string, saves it in its database and finally sends this string to the adversary.
			
			\textit {\textbf {Random Oracle}} \textsf{O}$(F_W)$.
			The value of $F_W$ is given to the random oracle \textsf{O}$(\cdot)$ as input to generate $K_W$ which is the hash value of $F_W$. Upon receiving a new query, the challenger first searches its database, if it finds a tuple matching that value, it will return the respective $K_W$ to the adversary. If, however, there is no matching tuple in its database, it will choose a random $\lambda$-bit string $K_W$ and returns it to the adversary. Then it saves the tuple $\langle F_W$, $K_W \rangle $ into its database.\\
			
			\item \textit{\textbf{KeyGen.}} The challenger generates a list $list$, which is empty at bigining, and generates $q_c$ sets of public keys $W_1,\dots, W_{q_c}$, of size $s$ each, as below.\\
			For generating each set $W_j$, the challenger chooses $2s$ $m$-bit random integers $b_{i, j}$ and $r_{i, j}$, for $i \in \{1,\dots, s\}$. Then, it generates tuples $\langle e_{1, j}, d_{1, j} \rangle$,\dots, $\langle e_{s, j}, d_{s, j} \rangle$ as below, and saves them in its list.
			\begin{align}
			e_{i, j}&=r_{i, j}e_r + b_{i, j}e_{\alpha} \nonumber
			\\& = p(r_{i, j}+y_\alpha  b_{i, j})+zq(k_\alpha b_{i, j}+k_r r_{i, j}) \nonumber
			\\& = y\acute{}_{i, j} p+k\acute{}_{i, j} zq,\\
			d_{i, j}&=(g^{p})^{r_{i, j}}(g^{py_\alpha})^{b_{i, j}} \bmod N \nonumber
			\\& = g^{p(b_{i, j} y_\alpha +r_{i, j})} \bmod N.
			\end{align}
			where, $k\acute{}_{i, j}$ and $y\acute{}_{i, j}$ are some unknown random integers.\\            
			Finally, the challenger includes the public keys $e_{1, j}$,\dots,  $e_{s, j}$ in $W_j$ and sends the $q_c+1$ sets $\langle W_1,\dots, W_{q_c}\rangle$ and $ W_c$ in a random order to the adversary. Note that the adversary is not able to distinguish $W_c$ from other sets. 
			
			\item \textit{\textbf{Phase 1.}} The adversary is allowed to ask $q_c$ queries from \textsf{Corrupt} oracle as well as one query from the \textsf{Test} oracle. These oracles are controlled by the challenger $B$.
			\begin{itemize}
				\item \textsf{\textit{\textbf{Corrupt}($e_i$).}} Upon receiving the \textsf{Corrupt}($e_{i}$) query, where $e_{i} \in W_1\cup  W_2 \dots, W_{q_c}\cup W_c$, the challenger searches its list $list$ to find the respective \color{black} set \color{black} $W_j$, where $e_{i} \in W_j$. If $W_j = W_c$ rejects the query and the game is finished (because in this case, $B$ is unable to compute the corresponding private key). Otherwise, the challenger  retrieves $\langle e_i, d_i \rangle$ from $list$ and sends $d_{i}$ back to the adversary.
				
				\item \textsf{\textit{\textbf{Test}($W^*$).}} When the adversary decides to end Phase 1, it chooses one of $q_c+1$ sets of public keys, such that for all $e_i \in W^*$, $e_i$ must appear in none of the \textsf{Corrupt}($e_i$) queries. Then the adversary sends \textsf{Test}($W^*$) to the challenger. After receiving this query, the challenger searches its list $list$ to find respective set $W_j$. If $W_j \neq W_c$ it sets $b=0$, reject the query and terminates the game. Since it means that challenger can compute the shared key $K_{W^*}$ by itself. 
				\onecolumn 
				\begin{table}
					\label{TableComp}
					\caption{ Comparison of NIKE schemes in terms of security and efficiency}
					\begin{center}
						%35 55 40 55
						\begin{tabular}{|p{90pt}|p{70pt}|p{85 pt}|p{110pt}|p{92pt}|}
							\hline
							\centering\textbf{Scheme} & \centering\textbf{Security} & \centering\textbf{Computational  complexity} & \centering\textbf{Based on} &  \color{black} \textbf{Memory complexity} \\
							
							\hline
							Interactive Two-Party key agreement \cite{X509} & Secure  & 3.29 s in single-hop network & ECC-implicit X.509 certificates$^{+}$& $O(1)$ \\
							
							\hline
							Interactive Two-Party key agreement \cite{X509} & Secure  & 10.41 s in single-hop network & ECDSA-explicit X.509 certificates, DER format$^{++}$&   $O(1)$ \\
							
							\hline
							Interactive Two-Party key agreement \cite{X509} & Secure  & 14.05 s in single-hop network & ECDSA-explicit X.509 certificates, PEM format $^{++}$&  $O(1)$  \\
							
							\hline
							GGH \cite{5} & Not proven \& broken by \cite{13} \  & More than 33 s for 19 parties*** & Lattice&  $O(\lambda^5 \log(\lambda))$ $\dagger$  \\
							
							\hline	
							GGHLite	\cite{10}, \cite{11} & Not proven \& broken by \cite{13}& 1.75 s for 7 parties on 16-core CPU & Lattice  &  $O(\lambda \log^2(\lambda))$  \\
							
							\hline
							GGH15 \cite{12} & Not proven \& broken by \cite{14} & More than 33 s for 19 parties*** & Lattice &  $\Omega(d^5\lambda^2\log^4(d\lambda))$ ($d$ is diameter of graph)\\
							
							\hline
							CLT13 \cite{15}	  & Not proven \& broken by \cite{16} & 134 s for 19 parties** & CRT &   $O(\lambda^2 \log(\lambda))$  \\
							\hline
							
							5Gen extension of CLT13 \cite{5Gen}	  & Not proven  & 33 s for 19 parties** & CRT &   $O(\lambda^2 \log(\lambda))$    \\
							
							\hline
							CLT15 \cite{CLT15}& Not proven \& broken by \cite{17} & More than 33 s for 19 parties*** & CRT  &   $O(\lambda^2 \log(\lambda))$  \\
							
							\hline
							Rao \cite{3}  & Secure in the standard model & More than 33 s for 19 parties* & $i\mathcal{O} $&  $O(poly(\lambda))$  \\
							
							\hline
							Yamakawa et al. \cite{7}  & Secure in the standard model & More than 33 s for 19 parties*& $i\mathcal{O}$ & $O(poly(\lambda))$  \\
							
							\hline
							Boneh and Zhandry \cite{8} & Secure in the standard model & More than 33 s for 19 parties*& $i\mathcal{O}$ &  $O(poly(\lambda))$   \\
							
							\hline
							Khurana et al. \cite{9} & Secure in the standard model & More than 33 s for 19 parties* & $i\mathcal{O}$  &  $O(poly(\lambda))$  \\
							
							\hline
							Ma and Zhandry \cite{MZ17}& Not proven \& not broken & More than 33 s for 19 parties***& CRT&  $O(\lambda^3)$   \\ 
							
							\hline
							Blundo et al. \cite{Blundo98}& $k$-secure in the standard model & More than 33 s for 19 parties*** & Symmetric polynomials & $\binom{k+t-1}{t-1}^{+++}$  \\
							
							\hline
							Eskeland \cite{2}& Not proven \& \textbf{broken in this paper} & 281 ms for 19 parties on CC2538 Chip& Root extraction in RSA modulus& $O(1)$   \\
							
							\hline
							Proposed scheme & \textbf{Secure in the random oracle model} & 281 ms for 19 parties on CC2538 Chip& Root extraction in RSA modulus & $O(1)$  \\
							\hline
						\end{tabular} \\\setlength\tabcolsep{1pt}
						\begin{flushleft} 
* The $i\mathcal{O}$-based schemes seem to be imperactical because multilinear maps are needed to construct $i\mathcal{O}$ \cite{5}, while multilinear maps can be used to construct MP-NIKE scheme by itself \cite{4}. \\
** Google Compute Engine servers with a 32-core CPU at 2.5 GHz, 208 GB RAM, and 100 GB disk storage.\\
*** The 5Gen extension of CLT13 is the most efficient multilinear map \cite{5Gen}, \cite{weakCLT13}, so other multilinear maps take more than 33 seconds.\\
$^{+}$ ECDH key exchange certified using the implicit Elliptic Curve Qu-Vanstone (ECQV) certificates \cite{ECQV} (Using CC2538 chip and IEEE 802.15.4e technology).\\
$^{++}$  ECDH key exchange with public coefficients certified using the ECC Digital Signature Algorithm (ECDSA) encoded through the standard Privacy Enhanced Mail (PEM) format and the binary Distinguished Encoding Rules (DER) format (Using CC2538 chip and IEEE 802.15.4e technology).\\
\color{black} $^{+++}$ $t$ is the size of conference and $d$ is the degree of resiliency. \\
$\dagger$ $\lambda$ is the security parameter.\color{black}
\end{flushleft}
					\end{center}
				\end{table}
				
				\twocolumn
				Based on the security model of Section \ref{SecModel}, it must send a shared key $K_{W^*}$ to the adversary, but since by playing the rest of the game the challenger will not learn any useful information from the adversary, the challenger finishes the game. Otherwise, if $W_j = W_c$, then the challenger sets $b=1$, which means it cannot compute $K_{W^*}$ by itself and it must send a random string to the adversary to employ the response of the adversary for computing $K_{W^*}$. So, the challenger generates a random string $\mathit{rand} \leftarrow \{0, 1\}^{\lambda}$ and sends it back to the adversary.
			\end{itemize}
			
			\item \textit{\textbf{Phase 2.}} This phase is the same as Phase 1, except that the adversary does not access the \textsf{Test} oracle and it is not allowed to send \textsf{Corrupt}($e_i$), where $e_i \in W^*$. Moreover, it is allowed to send $q_H$ queries, to the random oracle \textsf{O}, aggregately in Phase 1 and 2. 
			
			\item \textit{\textbf{Guess.}} In this phase, the adversary must guess a bit $b\acute{} \in \{0, 1\}$ and send it to the algorithm $B$. If $b\acute{}=b$ it means that the adversary has previously obtained $K_{W^*}$ by sending \textsf{O}($F_{W^*}$) to the random oracle. So, the tuple $ \langle F_{W^*}, K_{W^*} \rangle $ is available in the random oracle database. Now the challenger must determine it. 	
			The challenger checks that if ${F_i}^{e_{\alpha}}={d_{\alpha}}^{\prod\limits_{k: U_{c_k} \in W_c}e_{c_k}}$, for all tuples $ \langle F_i, K_i \rangle $ of database of random oracle \textsf{O}. Then it outputs $H(F_{W_i})$ as the shared key of the group $W_c$.\\
			
			\textbf{Analysis.} The game will be successfully terminated if the challenger rejects none of the adversary queries. The probability of rejecting none of the $q_c$ \textsf{Corrupt} queries is $\delta^{q_c}$, where $\delta = \frac{sq_c}{(q_{c}+1)s}=\frac{q_c}{q_{c}+1}$ and the probability of not rejecting the \textsf{Test} query is at least equal to $(\frac{1}{q_c+1})= (1-\delta)$. On the other hand, the probability of responding to all the $q_H$ random oracle queries is 1. Hence, the game will be successfully terminated with probability ${(\frac{q_c}{q_c+1})}^{q_c}(1-\frac{q_c}{q_c+1})$. 
			%The value of pr is maximized at $\delta_{opt}=\frac{q_c}{q_c+1}$.
			
			By definition \ref{ep} the advantage of the adversary $A$ to break our scheme is $Adv_A^\mathcal{E}=\Big| Pr[b\acute{}=b] - \frac{1}{2} \Big| \leq \epsilon$. So if the game is finished successfully, the algorithm $B$ will be able to obtain $K_{W^*}$ from the records of the random oracle database.
			It means that if the adversary is able to break our scheme with advantage $\epsilon$, then the challenger can break the security of the proposed scheme with only one public/private key pair with advantage $\epsilon$ ${(\frac{q_c}{q_c+1})}^{q_c}(1-\frac{q_c}{q_c+1})\approx \frac{\epsilon}{ e(q_c+1)}$ and with an additional time of about $q_H$ for retrieving $K_{W^*}$ among the $q_H$ records of the random oracle database.
		\end{enumerate}
	\end{proof} 
	
	The technique used in the proof of Theorem \ref{Th2} is the same as the technique of analysis of Boneh and Franklin IBE scheme \cite{BF}. \\ 
	
	Note that we have assumed that the bit length of the $e_i$'s in the real scheme and the above game are equal. The adversary is not able to distinguish between the public/private key pairs generated in the proposed scheme and the public/private key pairs generated in the above game, because algorithm $B$ selects the $r_{i,j}$'s and $b_{i,j}$'s at random. So in $e_{i,j}=p(r_{i,j}+y_\alpha b_{i,j})+zq(k_\alpha b_{i,j}+r_{i,j} k_r)$ the values of $(r_{i,j}+y_\alpha b_{i,j})$ and $(k_\alpha b_{i,j}+r_{i,j} k_r)$ are random integers. We can, therefore, be sure that the $(e_{i, j},d_{i, j})$'s look like uniformly distributed random numbers and the adversary cannot distinguish those from random integers.
	\section{Performance evaluation}
	\label{performance}
	\color{black}In this section, we compare the proposed scheme to the previous MP-NIKE proposals. The proposed scheme is efficient in computation time and also in terms of memory complexity. It is remarkable how it can work on microcontrollers like CC2538. The CC2538 is a wireless microcontroller System-on-Chip (SoC) with 32KB on-chip RAM and up to 512KB on-chip flash.
	\color{black}Based on the data-sheet of the Texas Instruments CC2538 \cite{cc2538} ch. 22, p.503, an RSA-CRT-1024 \color{black}(modular exponentiation by using CRT algorithm) \color{black}takes around 15.6 ms on this chip \cite{cc2538}. We can state that in our proposed scheme  for obtaining a shared key among nineteen users at 80-bit security, each user requires to compute 18 modular exponentiations in a 1024-bit modulus which takes about $18\times15.6=281$ milliseconds. For comparison, in the 5Gen multilinear map (extension of CLT13), which is the most efficient multilinear map \cite{weakCLT13,5Gen}, the 19-party key agreement in 80-bit security takes about 33 seconds and, \color{black} as shown in Table 2, its memory complexity is polynomial in security parameter \cite{5Gen}. By using the GGHLite multilinear map,   computing the shared key among seven users takes 1.75 seconds in a 16-core CPU \cite{10}. \color{black} 
	The CLT13 and GGHLite are the only multilinear maps which have been implemented so far. 
	Table 2 compares the previous proposals against the proposed scheme in terms of security, computational overhead and memory complexity. \color{black}In our proposed scheme, each entity stores only one public/private key and a public parameter, which is supposed to be 4096 bits aggregately in 80-bit security. As illustrated in Table 2, among the non-interactive key exchange schemes, only the Eskeland's scheme and our proposed scheme have $O(1)$ of memory overhead.\color{black}
	\section{Conclusion}
	\label{Conclu}
	
	The key distribution problem is linked to symmetric cryptography approaches and is critical in environments with a large number of users or where users are changing \cite{Vijayakumar2018,Reegan2017}.  The MP-NIKE schemes aim to face with this issue efficiently.  In this paper, we present an attack against the Eskeland's MP-NIKE scheme and then we proposed a new MP-NIKE scheme which is secure against any number of colluding users.    
	
	In terms of performance, our proposal is efficient in its various phases, including setup, key generation, deriving a shared key and updating the shared key after adding a new user.  Besides, the practical applicability of our proposal is clear.  In particular, we can use the proposed scheme in various cryptographic applications like broadcast encryption and secure group communication for IoT and WSNs. 
	
	As showed and proven in Section \ref{Security},  our proposed scheme bases its security on the intractability of the Fiat-Naor problem, which is equal to the root extraction modulo a composite number. While the security of the proposed MP-NIKE scheme relies on the random oracle model, an attractive future work will be to introduce a new MP-NIKE proposal which is secure in the standard model. 
	
	Finally, we would like to highlight that cybersecurity issues from the full plethora of IoT devices may be the next big nightmare for information and communications technology security administrators. We hope this contribution will increase the security of these devices and solve some of their problems.

	\EOD
\end{document}